\documentclass[11pt,amssymb,amsfont,a4paper]{article}
\usepackage{latexsym,amssymb,amsmath,amsthm,color}
\usepackage{geometry}
\usepackage{url}
\usepackage{graphicx}
\usepackage[utf8]{inputenc}
\usepackage[numbers]{natbib}
\usepackage{authblk}
\usepackage{tikz}
\usetikzlibrary{shapes}
\usepackage{arydshln}

\theoremstyle{definition}
\newtheorem{definition}{Definition}

\theoremstyle{plain}
\newtheorem{theorem}{Theorem}
\newtheorem{proposition}[definition]{Proposition}
\newtheorem{lemma}[definition]{Lemma}
\newtheorem{remark}[definition]{Remark}
\newtheorem{corollary}[definition]{Corollary}

\title{Doubly and triply extended MSRD codes}
\author{Umberto Mart{\'i}nez-Pe\~{n}as \thanks{umberto.martinez@uva.es}}
\affil{IMUVa-Mathematics Research Institute,\\University of Valladolid, Spain}

\date{}

\begin{document}

\maketitle

\begin{abstract}
In this work, doubly extended linearized Reed--Solomon codes and triply extended Reed--Solomon codes are generalized. We obtain a general result in which we characterize when a multiply extended code for a general metric attains the Singleton bound. We then use this result to obtain several families of doubly extended and triply extended maximum sum-rank distance (MSRD) codes that include doubly extended linearized Reed--Solomon codes and triply extended Reed--Solomon codes as particular cases. To conclude, we discuss when these codes are one-weight codes.


\end{abstract}

\section{Introduction} \label{sec intro}

Let $ \mathbb{F}_q $ denote the finite field of size $ q $, and denote by $ \mathbb{F}_q^n $ and $ \mathbb{F}_q^{m \times n} $ the spaces of row vectors of length $ n $ and matrices of size $ m \times n $, respectively, over $ \mathbb{F}_q $, for positive integers $ m $ and $ n $. We also denote $ \mathbb{N} = \{ 0,1,2,\ldots \} $ and $ [n] = \{ 1,2, \ldots, n \} $ for a positive integer $ n $. The Hamming metric in $ \mathbb{F}_q^n $ is given by $ {\rm d}_H(\mathbf{c}, \mathbf{d}) = | \{ i \in [n] \mid c_i \neq d_i \} | $, for $ \mathbf{c}, \mathbf{d} \in \mathbb{F}_q^n $.

Doubly extended Reed--Solomon codes \cite[Sec. 5.3]{pless} \cite[Ch. 11, Sec. 5]{macwilliamsbook} are the linear codes in $ \mathbb{F}_q^{n+2} $ given by the generator matrix
\begin{equation*}
\left( \begin{array}{cccc|cc}
1 & 1 & \ldots & 1 & 1 & 0 \\
 a_1 & a_2 & \ldots & a_n & 0 & 0 \\
 a_1^2 & a_2^2 & \ldots & a_n^2 & 0 & 0 \\
\vdots & \vdots & \ddots & \vdots & \vdots & \vdots \\
 a_1^{k-2} & a_2^{k-2} & \ldots & a_n^{k-2} & 0 & 0 \\
 a_1^{k-1} & a_2^{k-1} & \ldots & a_n^{k-1} & 0 & 1 \\
\end{array} \right) \in \mathbb{F}_q^{k \times (n+2)},
\end{equation*}
for $ k \in [n] $ and distinct $ a_1, a_2, \ldots , a_n \in \mathbb{F}_q^* $ (hence $ n \leq q-1 $ and $ n+2 \leq q+1 $, where equalities may be attained). 
One may show by using conventional polynomial results that the doubly extended Reed--Solomon code above is maximum distance separable (MDS). See \cite[Th. 5.3.4]{pless}. In other words, it attains the Singleton bound for the Hamming metric. Furthermore, these codes may have length $ q+1 $, which is conjectured to be maximum for most values of the code dimension $ k $. This is the well-known \textit{MDS conjecture} (see \cite[Sec. 7.4]{pless}), which has been proven for $ q $ prime \cite{ball}. 
%

Recently, a generalization of this result was given in \cite{neri-oneweight} for the sum-rank metric, a metric that simultaneously generalizes the Hamming metric and the rank metric \cite{delsartebilinear, gabidulin, roth}. The generalization of Reed--Solomon codes to the sum-rank metric is called linearized Reed--Solomon codes, introduced in \cite{linearizedRS}, which are maximum sum-rank distance (MSRD) codes, i.e., they attain the Singleton bound for the sum-rank metric. More general families of linear MSRD codes exist \cite{generalMSRD, twisted}. The authors of \cite{neri-oneweight} introduced doubly extended linearized Reed--Solomon codes and showed, using geometric tools, that they are also MSRD. 

In this work, we show how one may extend codes attaining the Singleton bound for any metric given by a weight. The metric considered for the extended codes is obtained by adding Hamming-metric components, as was done for the sum-rank metric in \cite{neri-oneweight} (Section \ref{sec singleton sums}). In Section \ref{sec multiply extended}, we provide necessary and sufficient conditions for multiply extended codes to attain the Singleton bound based on the original codes. In Sections \ref{sec doubly extended MSRD} and \ref{sec triply extended}, we apply double and triple extensions, respectively, to the general MSRD codes obtained in \cite{generalMSRD}, which include linearized Reed--Solomon codes (and therefore classical Reed--Solomon codes and Gabidulin codes \cite{gabidulin, roth}). In Section \ref{sec negative sum-rank}, we study what happens when the extended portion is not considered with Hamming-metric components, but by considering the rank metric in the whole added block, and show that doubly extended codes are no longer MSRD in general. Finally, in Section \ref{sec one-weight}, we investigate when the obtained doubly and triply extended MSRD codes are one-weight codes.

\section{The Singleton bound for sums of metrics} \label{sec singleton sums}

In this manuscript, we consider metrics given by weights. Here, a weight function is a function $ {\rm wt} : \mathbb{F}_q^n \longrightarrow \mathbb{N} $ satisfying the following properties:
\begin{enumerate}
\item
$ {\rm wt}(\mathbf{c}) \geq 0 $ and it equals $ 0 $ if, and only if, $ \mathbf{c} = \mathbf{0} $, for all $ \mathbf{c} \in \mathbb{F}_q^n $.
\item
$ {\rm wt}(\lambda \mathbf{c}) = {\rm wt}(\mathbf{c}) $, for all $ \mathbf{c} \in \mathbb{F}_q^n $ and all $ \lambda \in \mathbb{F}_q^* $.
\item
$ {\rm wt}(\mathbf{c} + \mathbf{d}) \leq {\rm wt}(\mathbf{c}) + {\rm wt}(\mathbf{d}) $, for all $ \mathbf{c}, \mathbf{d} \in \mathbb{F}_q^n $.
\end{enumerate}
Its associated metric is the function $ {\rm d} : (\mathbb{F}_q^n)^2 \longrightarrow \mathbb{N} $ given by $ {\rm d}(\mathbf{c}, \mathbf{d}) = {\rm wt}(\mathbf{c} - \mathbf{d}) $, for $ \mathbf{c}, \mathbf{d} \in \mathbb{F}_q^n $. It is straightforward to prove that a metric given by a weight as above is indeed a metric (see \cite[Th. 1.4.1]{pless}). 

As usual, we define the minimum distance of a code $ \mathcal{C} \subseteq \mathbb{F}_q^n $ (a code is just a set) with respect to $ {\rm d} $ as
$$ {\rm d}(\mathcal{C}) = \min \{ {\rm d}(\mathbf{c}, \mathbf{d}) \mid \mathbf{c}, \mathbf{d} \in \mathcal{C}, \mathbf{c} \neq \mathbf{d} \}. $$
It is well-known that, if $ \mathcal{C} $ is linear (i.e., an $ \mathbb{F}_q $-linear subspace of $ \mathbb{F}_q^n $), then $ {\rm d}(\mathcal{C}) = \min \{ {\rm wt}(\mathbf{c}) \mid \mathbf{c} \in \mathcal{C} \setminus \{ \mathbf{0} \} \} $, where $ {\rm wt} $ is the weight giving the metric $ {\rm d} $.

We will say that a metric $ {\rm d} $ satisfies the Singleton bound if
\begin{equation}
{\rm d}(\mathcal{C}) \leq n - k + 1,
\label{eq singleton bound}
\end{equation}
where $ k = \log_q|\mathcal{C}| $, for any code $ \mathcal{C} \subseteq \mathbb{F}_q^n $. Any metric given by a weight that is upper bounded by the Hamming weight satisfies the Singleton bound. Many examples exist, including the Hamming metric itself, the rank metric \cite{delsartebilinear, gabidulin}, the sum-rank metric \cite{linearizedRS}, the cover metric \cite{roth} and the multi-cover metric \cite{multicover}, among others. 

Some of these metrics, e.g., the sum-rank metric, the multi-cover metric or the Hamming metric itself, are given by sums of other metrics. In general, given weights $ {\rm wt}_i $ in $ \mathbb{F}_q^{n_i} $, for $ i \in [\ell] $, we may define their sum as
$$ {\rm wt}_{{\rm sum}}(\mathbf{c}) = {\rm wt}_1(\mathbf{c}_1) + {\rm wt}_2(\mathbf{c}_2) + \cdots + {\rm wt}_\ell(\mathbf{c}_\ell), $$
for $ \mathbf{c} = (\mathbf{c}_1, \mathbf{c}_2, \ldots, \mathbf{c}_\ell) \in \mathbb{F}_q^n $, where $ n = n_1 + n_2 + \cdots + n_\ell $ and $ \mathbf{c}_i \in \mathbb{F}_q^{n_i} $, for $ i \in [\ell] $. Clearly, $ {\rm wt}_{{\rm sum}} $ is a weight. We denote similarly the corresponding associated metric. It is easy to see that $ {\rm d}_{{\rm sum}} $ satisfies the Singleton bound if so do the metrics $ {\rm d}_i $, for $ i \in [\ell] $.  

In the remainder of the manuscript, we will only consider metrics $ {\rm d}: (\mathbb{F}_q^n)^2 \longrightarrow \mathbb{N} $ given by weights and satisfying the Singleton bound (\ref{eq singleton bound}).

\section{Multiply extended codes} \label{sec multiply extended}


In this section, we give a definition of multiply extended codes for general metrics and show that they attain the Singleton bound if so do certain codes related to the original code and the metric is extended by adding a Hamming-metric component. In Sections \ref{sec doubly extended MSRD} and \ref{sec triply extended}, we will particularize these results to construct doubly and triply extended MSRD codes.

\begin{theorem} \label{th multiply extended}
Let $ \mathbf{g}_1, \mathbf{g}_2, \ldots , \mathbf{g}_k \in \mathbb{F}_q^n $ be linearly independent, and let $ t \in [k] $. Consider the $ k $-dimensional linear code $ \mathcal{C}_e \subseteq \mathbb{F}_q^{n+t} $ with generator matrix
$$ G_e = \left( \begin{array}{c|cccc}
\mathbf{g}_1 & 1 & 0 & \ldots & 0 \\
\mathbf{g}_2 & 0 & 1 & \ldots & 0 \\
\vdots & \vdots & \vdots & \ddots & \vdots \\
\mathbf{g}_t & 0 & 0 & \ldots & 1 \\
\hline
\mathbf{g}_{t+1} & 0 & 0 & \ldots & 0 \\
\vdots & \vdots & \vdots & \ddots & \vdots \\
\mathbf{g}_k & 0 & 0 & \ldots & 0
\end{array} \right) \in \mathbb{F}_q^{k \times (n+t)}. $$
Define also the linear codes $ \mathcal{C}_I = \langle \{ \mathbf{g}_i \mid i \in I \} \rangle + \langle \mathbf{g}_{t+1}, \ldots, \mathbf{g}_k \rangle $, and set $ d_I = {\rm d}(\mathcal{C}_I) $, for $ I \subseteq [t] $. Here $ \langle \cdot \rangle $ denotes linear span. Then it holds that $ {\rm d}_e (\mathcal{C}_e) = \min \{ d_I + |I| \mid I \subseteq [t] \} $, where the metric $ {\rm d}_e : (\mathbb{F}_q^{n+t})^2 \longrightarrow \mathbb{N} $ is given by
$$ {\rm d}_e( (\mathbf{c}_1, \mathbf{c}_2), (\mathbf{d}_1, \mathbf{d}_2)) = {\rm d}(\mathbf{c}_1,\mathbf{d}_1) + {\rm d}_H(\mathbf{c}_2,\mathbf{d}_2), $$
for $ \mathbf{c}_1 , \mathbf{d}_1 \in \mathbb{F}_q^n $ and $ \mathbf{c}_2,\mathbf{d}_2 \in \mathbb{F}_q^t $.
\end{theorem}
\begin{proof}
Let $ \mathbf{e}_1, \mathbf{e}_2, \ldots, \mathbf{e}_t \in \mathbb{F}_q^t $ denote the canonical basis. A codeword in $ \mathcal{C}_e $ is of the form
$$ \mathbf{c} = \left( \sum_{i \in I} \lambda_i \mathbf{g}_i + \sum_{j=t+1}^k \lambda_j \mathbf{g}_j, \sum_{i \in I} \lambda_i \mathbf{e}_i \right), $$
where $ I \subseteq [t] $, $ \lambda_i \in \mathbb{F}_q^* $, for $ i \in I $, and $ \lambda_j \in \mathbb{F}_q $, for $ j = t+1, \ldots, k $. Note that possibly $ I = \varnothing $. Since $ \lambda_i \neq 0 $ for $ i \in I $, we deduce that
\begin{equation*}
\begin{split}
{\rm wt}_e (\mathbf{c}) & = {\rm wt} \left( \sum_{i \in I} \lambda_i \mathbf{g}_i + \sum_{j=t+1}^k \lambda_j \mathbf{g}_j \right) + {\rm wt}_H \left( \sum_{i \in I} \lambda_i \mathbf{e}_i \right) \\
 & = {\rm wt} \left( \sum_{i \in I} \lambda_i \mathbf{g}_i + \sum_{j=t+1}^k \lambda_j \mathbf{g}_j \right) + |I| \\
 & \geq {\rm d}_I + |I|. 
\end{split}
\end{equation*}
Therefore, we have that $ {\rm d}_e(\mathcal{C}_e) \geq \min \{ d_I + |I| \mid I \subseteq [t] \} $. 

Now, consider a subset $ I \subseteq [t] $ and take $ \mathbf{d} = \sum_{i \in I} \lambda_i \mathbf{g}_i + \sum_{j=t+1}^k \lambda_j \mathbf{g}_j \in \mathcal{C}_I $ such that $ {\rm wt}(\mathbf{d}) = d_I $, where $ \lambda_i \in \mathbb{F}_q $ for $ i \in I \cup \{ t+1, \ldots, k \} $. If $ J \subseteq I $ is such that $ \mathbf{d} = \sum_{i \in J} \lambda_i \mathbf{g}_i + \sum_{j=t+1}^k \lambda_j \mathbf{g}_j $ (i.e., $ \lambda_i = 0 $ if $ i \in I \setminus J $), then $ \mathbf{d} \in \mathcal{C}_J $ and thus
$$ d_J \leq {\rm wt}(\mathbf{d}) = d_I \leq d_J. $$
Hence we also have $ {\rm wt}(\mathbf{d}) = d_J $. Thus there exist $ I \subseteq [t] $ and a codeword $ \mathbf{d} = \sum_{i \in I} \lambda_i \mathbf{g}_i + \sum_{j=t+1}^k \lambda_j \mathbf{g}_j $ with $ {\rm wt}(\mathbf{d}) = d_I $, $ \lambda_{t+1}, \ldots, \lambda_k \in \mathbb{F}_q $ and $ \lambda_i \in \mathbb{F}_q^* $, for all $ i \in I $. Therefore,
$$ {\rm wt}_e \left( \sum_{i \in I} \lambda_i \mathbf{g}_i + \sum_{j=t+1}^k \lambda_j \mathbf{g}_j , \sum_{i \in I} \lambda_i \mathbf{e}_i \right) = d_I + |I|. $$
Considering all of the subsets $ I \subseteq [t] $ such that $ d_I = {\rm wt}(\mathbf{d}) $ for some $ \mathbf{d} = \sum_{i \in J} \lambda_i \mathbf{g}_i + \sum_{j=t+1}^k \lambda_j \mathbf{g}_j $, with $ \lambda_i \in \mathbb{F}_q^* $ for $ i \in I $ and $ \lambda_j \in \mathbb{F}_q $ for $ j \in \{ t+1, \ldots, k \} $, we conclude that $ {\rm d}_e(\mathcal{C}_e) = \min \{ d_I + |I| \mid I \subseteq [t] \} $. 
\end{proof}

We now deduce the following result on multiply extended codes that attain the Singleton bound.

\begin{corollary} \label{cor multiply singleton}
With notation as in Theorem \ref{th multiply extended}, the code $ \mathcal{C}_e $ attains the Singleton bound for $ {\rm d}_e $ if, and only if, so do the codes $ \mathcal{C}_I $ for $ {\rm d} $, for all $ I \subseteq [t] $.
\end{corollary}
\begin{proof}
Note that $ \dim(\mathcal{C}_e) = k $ and $ \dim(\mathcal{C}_I) = k + |I| - t $, for $ I \subseteq [t] $. Hence $ \mathcal{C}_I $ attains the Singleton bound for $ {\rm d} $ if, and only if,
$$ d_I = n - (k+|I|-t) + 1 = (n+t) - k - |I| + 1 . $$ 
We also have that $ \mathcal{C}_e $ attains the Singleton bound if, and only if,
\begin{equation*}
\begin{split}
{\rm d}_e (\mathcal{C}_e) & = \min \{ d_I + |I| \mid I \subseteq [t] \} \\
& = (n+t) - k + 1 \\
& = \min \{ (n+t) - k - |I| + 1 + |I| \mid I \subseteq [t] \},
\end{split}
\end{equation*}
and the result follows.
\end{proof}

\begin{remark}
Setting $ t = k $ and $ {\rm d} = {\rm d}_H $ (i.e., $ {\rm d}_e = {\rm d}_H $), then Corollary \ref{cor multiply singleton} recovers the well-known characterization of systematic generator matrices of MDS codes from \cite[Ch. 11, Th. 8]{macwilliamsbook}. In other words, when $ t = k $ and $ {\rm d} = {\rm d}_H $, Corollary \ref{cor multiply singleton} states that $ \mathcal{C}_e $ is MDS if, and only if, every square submatrix of $ G $ is invertible, where $ G $ is the matrix whose rows are $ \mathbf{g}_1, \mathbf{g}_2, \ldots, \mathbf{g}_k \in \mathbb{F}_q^n $. Corollary \ref{cor multiply singleton} extends this result to any $ t \in [k] $ and any metric $ {\rm d} $ given by a weight satisfying the Singleton bound.
\end{remark}

Finally, we note that we have a lattice of linear codes $ \mathcal{C}_I \subseteq \mathbb{F}_q^n $, for $ I \subseteq [t] $, with respect to inclusions or, equivalently, unions and intersections, i.e., we have the following inclusion graph:
$$ \begin{array}{ccccc}
 & & \mathcal{C}_{I \cup J} & & \\
 & \nearrow &  & \nwarrow & \\
\mathcal{C}_I & & & & \mathcal{C}_J \\
 & \nwarrow &  & \nearrow & \\
 & & \mathcal{C}_{I \cap J} . & & 
\end{array} $$

By taking systematic generator matrices, we deduce that the existence of a linear code in $ \mathbb{F}_q^{n + t} $ attaining the Singleton bound for $ {\rm d}_e $ is equivalent to the existence of a lattice of linear codes $ \mathcal{C}_I \subseteq \mathbb{F}_q^n $, for $ I \subseteq [t] $, as above, attaining the Singleton bound for $ {\rm d} $. This property also holds for the dual codes, as stated in the following proposition. 
%
Here, we define the dual of a linear code $ \mathcal{C} \subseteq \mathbb{F}_q^n $ as usual: $ \mathcal{C}^\perp = \{ \mathbf{d} \in \mathbb{F}_q^n \mid \mathbf{c} \cdot \mathbf{d}^\intercal = 0, \forall \mathbf{c} \in \mathcal{C} \} $.

\begin{proposition} \label{prop dual lattice}
Let $ \mathcal{C}_I \subseteq \mathbb{F}_q^n $, for $ I \subseteq [t] $, be a family of linear codes such that the map $ I \mapsto \mathcal{C}_I $ is a lattice isomorphism. Define now the linear codes $ \mathcal{D}_I = (\mathcal{C}_{I^c})^\perp \subseteq \mathbb{F}_q^n $, for $ I \subseteq [t] $, where $ I^c = [t] \setminus I $ denotes the complement of $ I $ in $ [t] $. Then the map $ I \mapsto \mathcal{D}_I $ is also a lattice isomorphism. 
\end{proposition}
\begin{proof}
Simply notice that, for $ I,J \subseteq [t] $, we have
$$ \mathcal{D}_I + \mathcal{D}_J = (\mathcal{C}_{I^c})^\perp + (\mathcal{C}_{J^c})^\perp = (\mathcal{C}_{I^c} \cap \mathcal{C}_{J^c})^\perp = (\mathcal{C}_{I^c \cap J^c})^\perp = (\mathcal{C}_{(I \cup J)^c})^\perp = \mathcal{D}_{I \cup J}, $$
$$ \mathcal{D}_I \cap \mathcal{D}_J = (\mathcal{C}_{I^c})^\perp \cap (\mathcal{C}_{J^c})^\perp = (\mathcal{C}_{I^c} + \mathcal{C}_{J^c})^\perp = (\mathcal{C}_{I^c \cup J^c})^\perp = (\mathcal{C}_{(I \cap J)^c})^\perp = \mathcal{D}_{I \cap J}. $$
\end{proof}

Assume that $ {\rm d} $ is a metric such that a linear code attains the Singleton bound if, and only if, so does its dual code. In such a case, Proposition \ref{prop dual lattice} states that we do not need to check the conditions in Corollary \ref{cor multiply singleton} for both the primary and dual codes, but only for one of them. This is the case of the sum-rank metric \cite[Th. 5]{gsrws}, and thus of the Hamming and rank metrics in particular.

\section{Doubly extended MSRD codes} \label{sec doubly extended MSRD}

In this section, we generalize the construction of doubly extended linearized Reed--Solomon codes from \cite{neri-oneweight} to the general family of MSRD codes from \cite{generalMSRD}. Using Corollary \ref{cor multiply singleton}, we will show that such doubly extended MSRD codes are again MSRD. 

Recall that the sum-rank metric \cite{multishot} in $ \mathbb{F}_{q^m}^n $ over $ \mathbb{F}_q $ for the length partition $ (g,r) $ is defined as a sum of rank metrics, i.e., sum-rank weights are given by
$$ {\rm wt}_{SR}(\mathbf{c}) = \sum_{i=1}^g {\rm wt}_R \left( \mathbf{c}^{(i)} \right), $$
for $ \mathbf{c} = \left( \mathbf{c}^{(1)}, \mathbf{c}^{(2)}, \ldots, \mathbf{c}^{(g)} \right) \in \mathbb{F}_{q^m}^n $, where $ \mathbf{c}^{(i)} \in \mathbb{F}_{q^m}^r $, for $ i \in [g] $, and $ n = gr $. Recall that rank weights in $ \mathbb{F}_{q^m}^r $ are given by $ {\rm wt}_R(\mathbf{d}) = \dim_{\mathbb{F}_q}( \langle d_1, d_2, \ldots, d_r \rangle_{\mathbb{F}_q} ) $, for $ \mathbf{d} = (d_1, d_2, \ldots, d_r) \in \mathbb{F}_{q^m}^r $.

We now give the definition of extended Moore matrices from \cite[Def. 3.4]{generalMSRD}.

\begin{definition}[\textbf{Extended Moore matrices \cite{generalMSRD}}] \label{def moore matrix extended}
Fix positive integers $ \ell $ and $ \eta $. Let $ \mathbf{a} = ( a_1, a_2, \ldots, a_\ell ) \in (\mathbb{F}_{q^m}^*)^\ell $ be such that $ N_{\mathbb{F}_{q^m}/\mathbb{F}_q}(a_i) \neq N_{\mathbb{F}_{q^m}/\mathbb{F}_q}(a_j) $ if $ i \neq j $, where $ N_{\mathbb{F}_{q^m}/\mathbb{F}_q}(a) = a \cdot a^q \cdots a^{q^{m-1}} $, for $ a \in \mathbb{F}_{q^m} $. For any $ \boldsymbol\beta = ( \beta_1, \beta_2, \ldots, \beta_\eta ) \in \mathbb{F}_{q^m}^{\eta} $ and $ k \in [\ell \eta] $, we define the \textit{extended Moore matrix} $ M_k(\mathbf{a}, \boldsymbol\beta) \in \mathbb{F}_{q^m}^{k \times (\ell \eta)} $ by $ M_k(\mathbf{a}, \boldsymbol\beta) = $
\begin{equation*}
\left( \begin{array}{lll|c|lll}
\beta_1 & \ldots & \beta_\eta & \ldots & \beta_1 & \ldots & \beta_\eta \\
 \beta_1^q a_1 & \ldots & \beta_\eta^q a_1 & \ldots & \beta_1^q a_\ell & \ldots & \beta_\eta^q a_\ell \\
\beta_1^{q^2} a_1^{\frac{q^2-1}{q-1}} & \ldots & \beta_\eta^{q^2} a_1^{\frac{q^2-1}{q-1}} & \ldots &  \beta_1^{q^2} a_\ell^{\frac{q^2-1}{q-1}} & \ldots &  \beta_\eta^{q^2} a_\ell^{\frac{q^2-1}{q-1}} \\
\vdots & \ddots & \vdots & \ddots & \vdots & \ddots & \vdots \\
 \beta_1^{q^{k-1}} a_1^{\frac{q^{k-1}-1}{q-1}} & \ldots &  \beta_\eta^{q^{k-1}} a_1^{\frac{q^{k-1}-1}{q-1}} & \ldots &  \beta_1^{q^{k-1}} a_\ell^{\frac{q^{k-1}-1}{q-1}} & \ldots &  \beta_\eta^{q^{k-1}} a_\ell^{\frac{q^{k-1}-1}{q-1}} \\
\end{array} \right),
\end{equation*}
and we denote by $ \mathcal{C}_k(\mathbf{a}, \boldsymbol\beta) \subseteq \mathbb{F}_{q^m}^{\ell \eta} $ the $ k $-dimensional linear code generated by $ M_k(\mathbf{a}, \boldsymbol\beta) $ (i.e., the rows of $ M_k(\mathbf{a}, \boldsymbol\beta) $ generate the vector space  $ \mathcal{C}_k(\mathbf{a}, \boldsymbol\beta) $).
\end{definition}

The following result \cite[Th. 3.12]{generalMSRD} characterizes when a code with an extended Moore matrix as generator or parity-check matrix is MSRD. 

\begin{theorem}[\textbf{\cite{generalMSRD}}] \label{th extended moore matrix is msrd sufficient}
Let $ \mathbf{a} = ( a_1, a_2, \ldots, a_\ell) \in (\mathbb{F}_{q^m}^*)^\ell $ be as in Definition \ref{def moore matrix extended}. Let $ \boldsymbol\beta = (\beta_1, \beta_2, \ldots, \beta_{\mu r}) \in \mathbb{F}_{q^m}^{\mu r} $, for positive integers $ \mu $ and $ r $, and set $ g = \ell \mu $. Define the $ \mathbb{F}_q $-linear subspace
\begin{equation}
 \mathcal{H}_i = \left\langle \beta_{(i-1)r+1}, \beta_{(i-1)r+2}, \ldots, \beta_{ir}  \right\rangle_{\mathbb{F}_q} \subseteq \mathbb{F}_{q^m},
\label{eq subspace from h-lin indep}
\end{equation}
for $ i \in [\mu] $. Given $ k \in [gr] $, the code $ \mathcal{C}_k(\mathbf{a}, \boldsymbol\beta) $ from Definition \ref{def moore matrix extended} is MSRD  over $ \mathbb{F}_q $ for the length partition $ (g,r) $ if, and only if, the following two conditions hold for all $ i \in [\mu] $:
\begin{enumerate}
\item
$ \dim_{\mathbb{F}_q}(\mathcal{H}_i) = r $, and
\item
$ \mathcal{H}_i \cap \left( \sum_{j \in \Gamma} \mathcal{H}_j \right) = \{ 0 \} $, for any set $ \Gamma \subseteq [\mu] $, such that $ i \notin \Gamma $ and $ |\Gamma| \leq \min \{ k,\mu \} -1 $.
\end{enumerate}
\end{theorem}

Several constructions of MSRD codes based on Theorem \ref{th extended moore matrix is msrd sufficient} were obtained in \cite{generalMSRD}. These include linearized Reed--Solomon codes \cite{linearizedRS} by taking $ \mu = 1 $ (in that case, Condition 2 is empty and Condition 1 means that $ \beta_1, \beta_2, \ldots, \beta_r $ are $ \mathbb{F}_q $-linearly independent). 

For our purposes, we also need to consider the $ k $-dimensional linear codes $ \mathcal{D}_k(\mathbf{a}, \boldsymbol\beta) \subseteq \mathbb{F}_{q^m}^{\ell \eta} $ with generator matrices $ M^\prime_k(\mathbf{a}, \boldsymbol\beta) = $
\begin{equation*}
\left( \begin{array}{lll|c|lll}
 \beta_1^q a_1 & \ldots & \beta_\eta^q a_1 & \ldots & \beta_1^q a_\ell & \ldots & \beta_\eta^q a_\ell \\
\beta_1^{q^2} a_1^{\frac{q^2-1}{q-1}} & \ldots & \beta_\eta^{q^2} a_1^{\frac{q^2-1}{q-1}} & \ldots &  \beta_1^{q^2} a_\ell^{\frac{q^2-1}{q-1}} & \ldots &  \beta_\eta^{q^2} a_\ell^{\frac{q^2-1}{q-1}} \\
\vdots & \ddots & \vdots & \ddots & \vdots & \ddots & \vdots \\
 \beta_1^{q^k} a_1^{\frac{q^k-1}{q-1}} & \ldots &  \beta_\eta^{q^k} a_1^{\frac{q^k-1}{q-1}} & \ldots &  \beta_1^{q^k} a_\ell^{\frac{q^k-1}{q-1}} & \ldots &  \beta_\eta^{q^k} a_\ell^{\frac{q^k-1}{q-1}} \\
\end{array} \right),
\end{equation*}
for $ k \in [\ell \eta] $. Observe that we have the following inclusion graph:
$$ \begin{array}{ccccc}
 & & \mathcal{C}_k(\mathbf{a}, \boldsymbol\beta) & & \\
 & \nearrow &  & \nwarrow & \\
\mathcal{C}_{k-1}(\mathbf{a}, \boldsymbol\beta) & & & & \mathcal{D}_{k-1}(\mathbf{a}, \boldsymbol\beta) \\
 & \nwarrow &  & \nearrow & \\
 & & \mathcal{D}_{k-2}(\mathbf{a}, \boldsymbol\beta) . & & 
\end{array} $$
The codes $ \mathcal{C}_k(\mathbf{a}, \boldsymbol\beta) $ are MSRD given Conditions 1 and 2 in Theorem \ref{th extended moore matrix is msrd sufficient}. We now show that the same conditions turn the codes $ \mathcal{D}_k(\mathbf{a}, \boldsymbol\beta) $ into MSRD codes.

\begin{lemma} \label{lemma for doubly extended lin RS codes}
Let $ \ell $, $ \mu $ and $ r $ be positive integers, let $ \mathbf{a} = ( a_1, a_2, \ldots, a_\ell) \in (\mathbb{F}_{q^m}^*)^\ell $ and $ \boldsymbol\beta = (\beta_1, \beta_2, \ldots, \beta_{\mu r}) \in \mathbb{F}_{q^m}^{\mu r} $ as in Theorem \ref{th extended moore matrix is msrd sufficient}, and set $ g = \ell \mu $. For $ k \in [g r] $, $ \mathcal{C}_k(\mathbf{a}, \boldsymbol\beta) $ is MSRD if, and only if, so is $ \mathcal{D}_k(\mathbf{a}, \boldsymbol\beta) $, in both cases over $ \mathbb{F}_q $ for the length partition $ (g,r) $.
\end{lemma}
\begin{proof}
For $ a, \beta \in \mathbb{F}_{q^m} $ and a positive integer $ i $, we have that
$$ \beta^{q^i} a^{\frac{q^i-1}{q-1}} = \beta^{q^i} a^{q^{i-1}} \cdots a^q \cdot a = \left( \beta^{q^{i-1}} a^{q^{i-2}} \cdots a^q \cdot a \right)^q a = \left( \beta^{q^{i-1}} a^{\frac{q^{i-2}-1}{q-1}} \right)^q a . $$
Hence it holds that
$$ M^\prime_k(\mathbf{a}, \boldsymbol\beta) = M_k(\mathbf{a}, \boldsymbol\beta)^q {\rm diag}(a_1, \ldots, a_1 | \ldots | a_\ell , \ldots , a_\ell), $$
where $ M_k(\mathbf{a}, \boldsymbol\beta)^q $ means that we raise every entry of $ M_k(\mathbf{a}, \boldsymbol\beta) $ to the $ q $th power, and $ {\rm diag}(\cdot) $ denotes diagonal matrix. In particular, the same holds for the corresponding codes, i.e., 
$$ \mathcal{D}_k(\mathbf{a}, \boldsymbol\beta) = \mathcal{C}_k(\mathbf{a}, \boldsymbol\beta)^q {\rm diag}(a_1, \ldots, a_1 | \ldots | a_\ell , \ldots , a_\ell), $$
where $ \mathcal{C}_k(\mathbf{a}, \boldsymbol\beta)^q $ means that we raise every component of every codeword of $ \mathcal{C}_k(\mathbf{a}, \boldsymbol\beta) $ to the $ q $th power. Now, observe that the map $ \phi : \mathbb{F}_{q^m}^{gr} \longrightarrow \mathbb{F}_{q^m}^{gr} $ given by
$$ \phi \left( c_1 , \ldots, c_{\mu r} | \ldots | c_{(\ell - 1)(\mu r) + 1}, \ldots , c_{\ell (\mu r)} \right) = \left( c_1^q a_1 , \ldots, c_{\mu r}^q a_1 | \ldots | c_{(\ell - 1)(\mu r) + 1}^q a_\ell, \ldots , c_{\ell (\mu r)}^q a_\ell \right) $$
is a semilinear isometry for the sum-rank metric over $ \mathbb{F}_q $ for the length partition $ (g,r) $, since $ a_i \neq 0 $, for $ i \in [\ell] $ (see \cite[Cor. 3.8]{sum-rank-products}). Hence the result follows.
\end{proof}

Therefore, we are in the situation of Corollary \ref{cor multiply singleton} for the sum-rank metric. For this reason, we define the following codes.

\begin{definition}
Let $ \mathbf{a} = ( a_1, a_2, \ldots, a_\ell ) \in (\mathbb{F}_{q^m}^*)^\ell $ be as in Definition \ref{def moore matrix extended}. Let $ \boldsymbol\beta = ( \beta_1, \beta_2, \ldots, \beta_\eta ) \in \mathbb{F}_{q^m}^{\eta} $ be arbitrary, for a positive integer $ \eta $. For $ k = 2,3, \ldots, \ell \eta $, we define the \textit{doubly extended Moore matrix} $ M^e_k(\mathbf{a}, \boldsymbol\beta) \in \mathbb{F}_{q^m}^{k \times (\ell \eta + 2)} $ by $ M^e_k(\mathbf{a}, \boldsymbol\beta) = $
\begin{equation*}
\left( \begin{array}{lll|c|lll|cc}
\beta_1 & \ldots & \beta_\eta & \ldots & \beta_1 & \ldots & \beta_\eta & 1 & 0 \\
 \beta_1^q a_1 & \ldots & \beta_\eta^q a_1 & \ldots & \beta_1^q a_\ell & \ldots & \beta_\eta^q a_\ell & 0 & 0 \\
\beta_1^{q^2} a_1^{\frac{q^2-1}{q-1}} & \ldots & \beta_\eta^{q^2} a_1^{\frac{q^2-1}{q-1}} & \ldots &  \beta_1^{q^2} a_\ell^{\frac{q^2-1}{q-1}} & \ldots &  \beta_\eta^{q^2} a_\ell^{\frac{q^2-1}{q-1}} & 0 & 0 \\
\vdots & \ddots & \vdots & \ddots & \vdots & \ddots & \vdots & \vdots & \vdots \\
 \beta_1^{q^{k-1}} a_1^{\frac{q^{k-1}-1}{q-1}} & \ldots &  \beta_\eta^{q^{k-1}} a_1^{\frac{q^{k-1}-1}{q-1}} & \ldots &  \beta_1^{q^{k-1}} a_\ell^{\frac{q^{k-1}-1}{q-1}} & \ldots &  \beta_\eta^{q^{k-1}} a_\ell^{\frac{q^{k-1}-1}{q-1}} & 0 & 1 \\
\end{array} \right),
\end{equation*}
and we denote by $ \mathcal{C}^e_k(\mathbf{a}, \boldsymbol\beta) \subseteq \mathbb{F}_{q^m}^{\ell \eta + 2} $ the $ k $-dimensional linear code generated by $ M^e_k(\mathbf{a}, \boldsymbol\beta) $.
\end{definition}

Thus, by Corollary \ref{cor multiply singleton} and Lemma \ref{lemma for doubly extended lin RS codes}, we deduce the following.

\begin{corollary} \label{cor doubly extended msrd codes}
Let $\ell $, $ \mu $ and $ r $ be positive integers, define $ g = \ell \mu $ and $ n = gr $, and let $ \mathbf{a} = ( a_1, a_2, \ldots, a_\ell) \in (\mathbb{F}_{q^m}^*)^\ell $ and $ \boldsymbol\beta = (\beta_1, \beta_2, \ldots, \beta_{\mu r}) \in \mathbb{F}_{q^m}^{\mu r} $ as in Theorem \ref{th extended moore matrix is msrd sufficient}. For $ k = 2,3, \ldots, n $, $ \mathcal{C}_k(\mathbf{a}, \boldsymbol\beta) \subseteq \mathbb{F}_{q^m}^n $ is MSRD (i.e., Conditions 1 and 2 in Theorem \ref{th extended moore matrix is msrd sufficient} hold) if, and only if, $ \mathcal{C}^e_k(\mathbf{a}, \boldsymbol\beta) \subseteq \mathbb{F}_{q^m}^{n + 2} $ is MSRD for the extended sum-rank metric
$$ {\rm d}_e( (\mathbf{c}, c_{n+1}, c_{n+2}), (\mathbf{d}, d_{n+1}, d_{n+2})) = {\rm d}_{SR}(\mathbf{c},\mathbf{d}) + {\rm d}_H((c_{n+1}, c_{n+2}), (d_{n+1}, d_{n+2})), $$
for $ \mathbf{c} , \mathbf{d} \in \mathbb{F}_{q^m}^n $ and $ c_{n+1}, c_{n+2}, d_{n+1}, d_{n+2} \in \mathbb{F}_{q^m} $, where $ {\rm d}_{SR} $ denotes the sum-rank metric in $ \mathbb{F}_{q^m}^n $ over $ \mathbb{F}_q $ for the length partition $ (g,r) $. 
\end{corollary}

In particular, if $ \ell = q-1 $ and $ \boldsymbol\beta = (\beta_1, \beta_2, \ldots, $ $ \beta_{\mu r}) $ $ \in \mathbb{F}_{q^m}^{\mu r} $ satisfies Conditions 1 and 2 in Theorem \ref{th extended moore matrix is msrd sufficient}, then the doubly extended code $ \mathcal{C}^e_k(\mathbf{a}, \boldsymbol\beta) \subseteq \mathbb{F}_{q^m}^{n + 2} $ is MSRD as in the corollary above, where $ n = (q-1)\mu r $ and where we consider in $ \mathbb{F}_{q^m}^n $ the sum-rank metric over $ \mathbb{F}_q $ for the length partition $ (g,r) $, $ g = (q-1)\mu $. See \cite{generalMSRD} for seven concrete explicit families of MSRD codes constructed in this way. All of them can be doubly extended as mentioned in this paragraph while preserving their MSRD property. 

In particular, choosing $ \mu = 1 $, Corollary \ref{cor doubly extended msrd codes} recovers \cite[Th. 4.6]{neri-oneweight} as a particular case for linearized Reed--Solomon codes, which in turn recovers the classical result \cite[Th. 5.3.4]{pless} for classical Reed--Solomon codes.

\section{Triply extended MSRD codes} \label{sec triply extended}

In contrast to the case of doubly extended MSRD codes (Section \ref{sec doubly extended MSRD}), triply extended MSRD codes are not always MSRD, as we show in this section. We will only consider $ 3 $-dimensional codes. 

We start with cases where triple extension preserves the MSRD property. Notice that the case of ($ 3 $-dimensional) classical Reed--Solomon codes and the Hamming metric in characteristic $ 2 $ \cite[p. 326, Ch. 11, Th. 10]{macwilliamsbook} is recovered from the following theorem by taking $ m=\mu=r=1 $ and $ \beta_1 = 1 $.

\begin{theorem} \label{th triply extended LRS}
Let $ m $ be odd, let $ q $ be even, and set $ n = (q-1)\mu r $ for positive integers $ \mu $ and $ r $. Let $ \boldsymbol\beta = (\beta_1, \ldots, \beta_{\mu r}) \in \mathbb{F}_{q^m}^{\mu r} $ satisfy Conditions 1 and 2 in Theorem \ref{th extended moore matrix is msrd sufficient}. Let $ \mathbf{a} = ( a_1, a_2, \ldots, a_{q-1} ) \in (\mathbb{F}_{q^m}^*)^{q-1} $ be such that $ N_{\mathbb{F}_{q^m}/\mathbb{F}_q}(a_i) \neq N_{\mathbb{F}_{q^m}/\mathbb{F}_q}(a_j) $ if $ i \neq j $. The triply extended code $ \mathcal{C}_e \subseteq \mathbb{F}_{q^m}^{n + 3} $ with generator matrix
$$ G_e = \left( \begin{array}{ccc|c|ccc|ccc}
\beta_1 & \ldots & \beta_{\mu r} & \ldots & \beta_1 & \ldots & \beta_{\mu r} & 1 & 0 & 0 \\
a_1 \beta_1^q & \ldots & a_1 \beta_{\mu r}^q & \ldots & a_{q-1} \beta_1^q & \ldots & a_{q-1} \beta_{\mu r}^q & 0 & 1 & 0 \\
a_1^{q+1} \beta_1^{q^2} & \ldots & a_1^{q+1} \beta_{\mu r}^{q^2} & \ldots & a_{q-1}^{q+1} \beta_1^{q^2} & \ldots & a_{q-1}^{q+1} \beta_{\mu r}^{q^2} & 0 & 0 & 1
\end{array} \right) \in \mathbb{F}_{q^m}^{3 \times (n + 3)} $$
is MSRD for the extended sum-rank metric
$$ {\rm d}_e( (\mathbf{c}, \mathbf{c}^\prime), (\mathbf{d}, \mathbf{d}^\prime)) = {\rm d}_{SR}(\mathbf{c},\mathbf{d}) + {\rm d}_H(\mathbf{c}^\prime, \mathbf{d}^\prime), $$
for $ \mathbf{c} , \mathbf{d} \in \mathbb{F}_{q^m}^n $ and $ \mathbf{c}^\prime, \mathbf{d}^\prime \in \mathbb{F}_{q^m}^3 $, where $ {\rm d}_{SR} $ denotes the sum-rank metric in $ \mathbb{F}_{q^m}^n $ over $ \mathbb{F}_q $ for the length partition $ (g, r) $, where $ g = (q-1) \mu $.
\end{theorem}
\begin{proof}
By Corollary \ref{cor multiply singleton} and Lemma \ref{lemma for doubly extended lin RS codes}, we only need to show that the code with generator matrix
$$ G = \left( \begin{array}{ccc|c|ccc}
\beta_1 & \ldots & \beta_{\mu r} & \ldots & \beta_1 & \ldots & \beta_{\mu r} \\
a_1^{q+1} \beta_1^{q^2} & \ldots & a_1^{q+1} \beta_{\mu r}^{q^2} & \ldots & a_{q-1}^{q+1} \beta_1^{q^2} & \ldots & a_{q-1}^{q+1} \beta_{\mu r}^{q^2} 
\end{array} \right) \in \mathbb{F}_{q^m}^{2 \times n} $$
is MSRD over $ \mathbb{F}_q $ for the length partition $ (g, r) $. 

First, since $ q $ is even, then if $ a,b \in \mathbb{F}_q $ are such that $ a \neq b $, then $ a^2 - b^2 = (a-b)^2 \neq 0 $, hence $ a^2 \neq b^2 $. Therefore if $ i \neq j $, since $ N_{\mathbb{F}_{q^m}/\mathbb{F}_q} (a_i) \neq N_{\mathbb{F}_{q^m}/\mathbb{F}_q} (a_j) $, we deduce that
$$ N_{\mathbb{F}_{q^m}/\mathbb{F}_q} (a_i^{q+1}) = N_{\mathbb{F}_{q^m}/\mathbb{F}_q} (a_i)^2 \neq N_{\mathbb{F}_{q^m}/\mathbb{F}_q} (a_i)^2 = N_{\mathbb{F}_{q^m}/\mathbb{F}_q} (a_j^{q+1}). $$ 

Second, since $ m $ is odd, then $ \tau : \mathbb{F}_{q^m} \longrightarrow \mathbb{F}_{q^m} $ given by $ \tau (a) = a^{q^2} $, for $ a \in \mathbb{F}_{q^m} $, is a field automorphism such that $ \{ a \in \mathbb{F}_{q^m} \mid a^{q^2} = a \} = \mathbb{F}_q $. In particular, $ N_{\mathbb{F}_{q^m}/\mathbb{F}_q}(a) = a \tau(a) \cdots \tau^{m-1}(a) $, for $ a \in \mathbb{F}_{q^m} $. Hence the generator matrix $ G $ is an extended Moore matrix (Definition \ref{def moore matrix extended}) satisfying the conditions in Theorem \ref{th extended moore matrix is msrd sufficient}, and therefore the code it generates is MSRD and we are done. 
\end{proof}

On the other hand, when $ m $ is even or $ q $ is odd, a triply extended (full-length) linearized Reed--Solomon code is never MSRD.

\begin{proposition} \label{prop triply extended never msrd}
Let $ \beta_1, \beta_2 , \ldots, \beta_m \in \mathbb{F}_{q^m} $ be $ \mathbb{F}_q $-linearly independent and let $ \mathbf{a} = ( a_1, a_2, $ $ \ldots, $ $ a_{q-1} ) \in (\mathbb{F}_{q^m}^*)^{q-1} $ be such that $ N_{\mathbb{F}_{q^m}/\mathbb{F}_q}(a_i) \neq N_{\mathbb{F}_{q^m}/\mathbb{F}_q}(a_j) $ if $ i \neq j $. Set $ n = (q-1)m $. If $ m $ is even or $ q $ is odd, then the triply extended code $ \mathcal{C}_e \subseteq \mathbb{F}_{q^m}^{n + 3} $ with generator matrix
$$ G_e = \left( \begin{array}{ccc|c|ccc|ccc}
\beta_1 & \ldots & \beta_m & \ldots & \beta_1 & \ldots & \beta_m & 1 & 0 & 0 \\
a_1 \beta_1^q & \ldots & a_1 \beta_m^q & \ldots & a_{q-1} \beta_1^q & \ldots & a_{q-1} \beta_m^q & 0 & 1 & 0 \\
a_1^{q+1} \beta_1^{q^2} & \ldots & a_1^{q+1} \beta_m^{q^2} & \ldots & a_{q-1}^{q+1} \beta_1^{q^2} & \ldots & a_{q-1}^{q+1} \beta_m^{q^2} & 0 & 0 & 1
\end{array} \right) \in \mathbb{F}_{q^m}^{3 \times (n + 3)} $$
is not MSRD for the extended sum-rank metric
$$ {\rm d}_e( (\mathbf{c}, \mathbf{c}^\prime), (\mathbf{d}, \mathbf{d}^\prime)) = {\rm d}_{SR}(\mathbf{c},\mathbf{d}) + {\rm d}_H(\mathbf{c}^\prime, \mathbf{d}^\prime), $$
for $ \mathbf{c} , \mathbf{d} \in \mathbb{F}_{q^m}^n $ and $ \mathbf{c}^\prime, \mathbf{d}^\prime \in \mathbb{F}_{q^m}^3 $, where $ {\rm d}_{SR} $ denotes the sum-rank metric in $ \mathbb{F}_{q^m}^n $ over $ \mathbb{F}_q $ for the length partition $ (q-1, m) $.
\end{proposition}
\begin{proof}
We first consider the case where $ m $ is even. Since $ \mathbb{F}_{q^2} \subseteq \mathbb{F}_{q^m} $ in this case, there exists an invertible matrix $ A \in \mathbb{F}_q^{m \times m} $ such that the first two components of $ (\beta_1, \beta_2, \ldots, \beta_m) A \in \mathbb{F}_{q^m}^m $ lie in $ \mathbb{F}_{q^2} $. Since such a multiplication constitutes a linear sum-rank isometry, we may assume that $ \beta_1, \beta_2 \in \mathbb{F}_{q^2} $ without loss of generality. Let 
\begin{equation}
G = \left( \begin{array}{ccc|c|ccc}
\beta_1 & \ldots & \beta_m & \ldots & \beta_1 & \ldots & \beta_m \\
a_1^{q+1} \beta_1^{q^2} & \ldots & a_1^{q+1} \beta_m^{q^2} & \ldots & a_{q-1}^{q+1} \beta_1^{q^2} & \ldots & a_{q-1}^{q+1} \beta_m^{q^2} 
\end{array} \right) \in \mathbb{F}_{q^m}^{2 \times n}.
\label{eq generator for proof triply not msrd}
\end{equation}
Since $ \beta_i - \beta_i^{q^2} = 0 $ ($ \beta_i \in \mathbb{F}_{q^2} $), for $ i = 1,2 $, we conclude that the codeword $ ( a_1^{q+1}, -1 ) G $ has sum-rank weight at most $ n-2 $, hence the code generated by $ G $ is not MSRD over $ \mathbb{F}_q $ for the length partition $ (q-1,m) $. Thus the code generated by $ G_e $ is not MSRD with respect to $ {\rm d}_e $ by Corollary \ref{cor multiply singleton}. 

We now consider the case where both $ q $ and $ m $ are odd. By assumption, we have that $ \{ N_{\mathbb{F}_{q^m}/\mathbb{F}_q}(a_i) \mid i \in [q-1] \} = \mathbb{F}_q^* $. Since $ q $ is odd, there exist $ 1 \leq i < j \leq q-1 $ such that $ N_{\mathbb{F}_{q^m}/\mathbb{F}_q}(a_i) = - N_{\mathbb{F}_{q^m}/\mathbb{F}_q}(a_j) $. In particular, 
$$ N_{\mathbb{F}_{q^m}/\mathbb{F}_q} (a_i^{q+1}) = N_{\mathbb{F}_{q^m}/\mathbb{F}_q} (a_i)^2 = N_{\mathbb{F}_{q^m}/\mathbb{F}_q} (a_i)^2 = N_{\mathbb{F}_{q^m}/\mathbb{F}_q} (a_j^{q+1}). $$ 
Consider the matrix $ G $ as in (\ref{eq generator for proof triply not msrd}). Since $ N_{\mathbb{F}_{q^m}/\mathbb{F}_q} (a_i^{q+1}) = N_{\mathbb{F}_{q^m}/\mathbb{F}_q} (a_j^{q+1}) $ and $ 2 $ and $ m $ are coprime, there exists $ \beta \in \mathbb{F}_{q^m}^* $ such that $ a_i^{q+1} \beta = a_j^{q+1} \beta^{q^2} $ by Hilbert's Theorem 90 \cite[p. 288, Th. 6]{lang}. Now, there exist invertible matrices $ A_i, A_j \in \mathbb{F}_q^{m \times m} $ such that $ 1 $ is the first component of $ (\beta_1, \beta_2, \ldots, \beta_m) A_i $ and $ \beta $ is the first component of $ (\beta_1, \beta_2, \ldots, \beta_m) A_j $. Let $ A_l = I_m $ for $ l \in [q-1] \setminus \{ i,j \} $. Denoting 
$$ {\rm diag}\left(A_1,A_2, \ldots, A_{q-1} \right) = \left( \begin{array}{cccc}
A_1 & 0 & \ldots & 0 \\
0 & A_2 & \ldots & 0 \\
\vdots & \vdots & \ddots & \vdots \\
0 & 0 & \ldots & A_{q-1}
\end{array} \right) \in \mathbb{F}_q^{n \times n} , $$
we deduce that $ G \cdot {\rm diag}\left(A_1,A_2, \ldots, A_{q-1} \right) $ contains the submatrix
$$ \left( \begin{array}{cc}
1 & \beta \\
a_i^{q+1} & a_j^{q+1} \beta^{q^2}
\end{array} \right), $$
which is not invertible since $ a_i^{q+1} \beta = a_j^{q+1} \beta^{q^2} $. Since multiplying by the invertible block diagonal matrix $ {\rm diag} \left(A_1,A_2, \ldots, A_{q-1} \right) \in \mathbb{F}_q^{n \times n} $ constitutes a linear sum-rank isometry, we deduce that the code generated by $ G $ is not MSRD over $ \mathbb{F}_q $ for the length partition $ (q-1,m) $. Thus the code generated by $ G_e $ is not MSRD with respect to $ {\rm d}_e $ by Corollary \ref{cor multiply singleton}. 
\end{proof} 

\begin{remark}
Notice that Proposition \ref{prop triply extended never msrd} works with the same proof in more general cases, where we consider the sum-rank metric in $ \mathbb{F}_{q^m}^n $, $ n =gr $, for the length partition $ (g,r) $, $ g = (q-1)\mu $, using $ \boldsymbol\beta = (\beta_1, \ldots, \beta_{\mu r}) \in \mathbb{F}_{q^m}^{\mu r} $ satisfying Conditions 1 and 2 in Theorem \ref{th extended moore matrix is msrd sufficient}, under the following assumptions: 1) $ m $ is even and $ \mathbb{F}_{q^2} \subseteq \mathcal{H}_i $ for some $ i \in [\mu] $; or 2) $ q $ and $ m $ are odd and $ \bigcup_{i=1}^\mu \mathcal{H}_i = \mathbb{F}_{q^m} $. Here, we define $ \mathcal{H}_i $, for $ i \in [\mu] $, as in (\ref{eq subspace from h-lin indep}). Since the linearized Reed--Solomon code case is $ \mu = 1 $, both conditions on the (single) subspace $ \mathcal{H}_1 $ hold when $ m $ is even or $ q $ and $ m $ are odd.
\end{remark}


\section{A negative result in the sum-rank metric} \label{sec negative sum-rank}

Up to this point, we have studied extensions of a metric $ {\rm d} $ by adding a Hamming-metric component $ {\rm d}_H $. The reader may wonder if the results in Section \ref{sec multiply extended} also hold if we extend $ {\rm d} $ by adding another metric, for instance, the rank metric. In this section, we give a negative answer to this question by trying to doubly extend MSRD codes as in Theorem \ref{th extended moore matrix is msrd sufficient} (for the largest value of $ \ell $, i.e., $ \ell = q-1 $) by adding a non-trivial rank-metric block and showing that the resulting code is not MSRD even if the conditions in Corollary \ref{cor multiply singleton} hold. 

\begin{proposition} \label{prop doubly ext MSRD not MSRD}
Let $ a_1, a_2, \ldots, a_{q-1} \in \mathbb{F}_{q^m}^* $ be such that $ N_{\mathbb{F}_{q^m}/\mathbb{F}_q}(a_i) \neq N_{\mathbb{F}_{q^m}/\mathbb{F}_q}(a_j) $ if $ i \neq j $. Let $ \boldsymbol\beta = (\beta_1, \beta_2, \ldots, \beta_{\mu r}) \in \mathbb{F}_{q^m}^{\mu r} $ and $ \mathcal{H}_i = \langle \beta_{(i-1)r+1} , \ldots, \beta_{ir} \rangle_{\mathbb{F}_q} \subseteq \mathbb{F}_{q^m} $, for $ i \in [\mu] $, satisfy Conditions 1 and 2 in Theorem \ref{th extended moore matrix is msrd sufficient}. Consider the extended sum-rank metric
$$ {\rm d}_e( (\mathbf{c}, c_{n+1}, c_{n+2}), (\mathbf{d}, d_{n+1}, d_{n+2})) = {\rm d}_{SR}(\mathbf{c},\mathbf{d}) + {\rm d}_R((c_{n+1}, c_{n+2}), (d_{n+1}, d_{n+2})), $$
for $ \mathbf{c} , \mathbf{d} \in \mathbb{F}_{q^m}^n $ and $ c_{n+1}, c_{n+2}, d_{n+1}, d_{n+2} \in \mathbb{F}_{q^m} $, where $ {\rm d}_{SR} $ denotes the sum-rank metric in $ \mathbb{F}_{q^m}^n $ over $ \mathbb{F}_q $ for the length partition $ (g,r) $, where $ g = (q-1)\mu $ and $ n = gr $. Let $ a,b,c,d \in \mathbb{F}_{q^m} $ with $ (0,0) \notin \{ (a,c), (b,d), (a,b), (c,d) \} $. Then the extended $ 2 $-dimensional code $ \mathcal{C}_e $ with generator matrix 
$$ G_e = \left( \begin{array}{ccc|ccc|c|ccc|cc}
\beta_1 & \ldots & \beta_{\mu r} & \beta_1 & \ldots & \beta_{\mu r} & \ldots & \beta_1 & \ldots & \beta_{\mu r} & a & c \\
a_1 \beta_1^q & \ldots & a_1 \beta_{\mu r}^q & a_2 \beta_1^q & \ldots & a_2 \beta_{\mu r}^q & \ldots & a_{q-1} \beta_1^q & \ldots & a_{q-1} \beta_{\mu r}^q & b & d
\end{array} \right) $$
is MSRD for $ {\rm d}_e $ if, and only if, 
$$ - \tau^{-1} \notin \bigcup_{i=1}^{q-1} \left\lbrace a_i \beta^{q-1} \left| \beta \in \bigcup_{j=1}^\mu \mathcal{H}_j \setminus \{ 0 \} \right. \right\rbrace , $$
for every $ \tau \in \mathbb{F}_{q^m}^* $ such that $ a + \tau b $ and $ c + \tau d $ are $ \mathbb{F}_q $-linearly dependent. In particular, if $ \bigcup_{j=1}^\mu \mathcal{H}_j = \mathbb{F}_{q^m} $, then $ \mathcal{C}_e $ is not MSRD for all $ a,b,c,d \in \mathbb{F}_{q^m} $.
\end{proposition}
\begin{proof}
First of all, the reader may verify that there exists $ \tau \in \mathbb{F}_{q^m}^* $ such that $ a + \tau b $ and $ c + \tau d $ are $ \mathbb{F}_q $-linearly dependent, since $ (0,0) \notin \{ (a,c), (b,d), (a,b), (c,d) \} $. 

Let $ \mathbf{g}_1, \mathbf{g}_2 \in \mathbb{F}_{q^m}^n $ be the first and second rows of $ G_e $, respectively, projected on the first $ n $ coordinates. If $ \tau \in \mathbb{F}_{q^m}^* $ is such that $ a + \tau b $ and $ c + \tau d $ are $ \mathbb{F}_q $-linearly independent, then we have that 
$$ {\rm wt}_e (\mathbf{g}_1 + \tau \mathbf{g}_2, a + \tau b, c + \tau d) \geq n+1. $$
Therefore $ \mathcal{C}_e $ is not MSRD if, and only if, $ {\rm wt}_{SR}(\mathbf{g}_1 + \tau \mathbf{g}_2) = n-1 $, for some $ \tau \in \mathbb{F}_{q^m}^* $ such that $ a + \tau b $ and $ c + \tau d $ are $ \mathbb{F}_q $-linearly dependent. Fix one such $ \tau $. We have $ {\rm wt}_{SR}(\mathbf{g}_1 + \tau \mathbf{g}_2) = n-1 $ if, and only if, there exist $ \lambda_1, \ldots, \lambda_r \in \mathbb{F}_q $, not all zero, such that
$$ \sum_{k=1}^r \lambda_k \beta_{(j-1)r+k} + \tau a_i \sum_{k=1}^r \lambda_k \beta_{(j-1)r+k}^q = 0, $$
for some $ j \in [\mu] $ and some $ i \in [q-1] $. Let $ \beta = \sum_{k=1}^r \lambda_k \beta_{(j-1)r+k} \in \mathcal{H}_j \setminus \{ 0 \} $. Then the equation above is simply $ -\tau^{-1} = a_i \beta^{q-1} $. This is possible for some $ i \in [q-1] $ and some $ \beta \in \mathcal{H}_j \setminus \{ 0 \} $ if, and only if, 
$$ -\tau^{-1} \in \bigcup_{i=1}^{q-1} \left\lbrace a_i \beta^{q-1} \left| \beta \in \bigcup_{j=1}^\mu \mathcal{H}_j \setminus \{0 \} \right. \right\rbrace , $$
and we are done.

Finally, assume that $ \bigcup_{j=1}^\mu \mathcal{H}_j = \mathbb{F}_{q^m} $. For $ \tau \in \mathbb{F}_{q^m}^* $, there exists $ i \in [q-1] $ such that $ N_{\mathbb{F}_{q^m}/\mathbb{F}_q}(-\tau^{-1}) $ $ = N_{\mathbb{F}_{q^m}/\mathbb{F}_q}(a_i) $. By Hilbert's Theorem 90, there exists $ \beta \in \mathbb{F}_{q^m}^* = \bigcup_{j=1}^\mu \mathcal{H}_j \setminus \{ 0 \} $ such that $ -\tau^{-1} = a_i \beta^{q-1} $ and we conclude that $ \mathcal{C}_e $ is not MSRD when $ \bigcup_{j=1}^\mu \mathcal{H}_j = \mathbb{F}_{q^m} $.
\end{proof}
%
%

In the case where $ (\beta_1,\beta_2, \ldots, \beta_{\mu r}) $ is constructed using field reduction (as in the following lemma, see also \cite[Sec. 4.1]{generalMSRD}), we have the following easy criterion to determine when $ \bigcup_{j=1}^\mu \mathcal{H}_j = \mathbb{F}_{q^m} $.

\begin{lemma} \label{lemma union condition for field reduction}
Let $ m = r \rho $, for positive integers $ r $ and $ \rho $, and let $ (\beta_{(j-1)r+1} , \ldots, \beta_{jr}) = \gamma_j (\alpha_1, \ldots, \alpha_r) $, for $ j \in [\mu] $, where $ \alpha_1, \ldots, \alpha_r \in \mathbb{F}_{q^r} $ are $ \mathbb{F}_q $-linearly independent, and $ \gamma_1, \ldots, \gamma_\mu \in \mathbb{F}_{q^m}^* $ are such that $ \gamma_i $ and $ \gamma_j $ are $ \mathbb{F}_{q^r} $-linearly independent if $ i \neq j $. Define $ \mathcal{H}_j = \langle \beta_{(j-1)r+1} , \ldots, \beta_{jr} \rangle_{\mathbb{F}_q} \subseteq \mathbb{F}_{q^m} $, for $ j \in [\mu] $. In this setting, we have $ \bigcup_{j=1}^\mu \mathcal{H}_j = \mathbb{F}_{q^m} $ if, and only if, $ \mu = (q^m-1) / (q^r-1) $.
\end{lemma}
\begin{proof}
In this case, the condition $ \bigcup_{j=1}^\mu \mathcal{H}_j = \mathbb{F}_{q^m} $ holds if, and only if, $ \{ [\gamma_1], \ldots, [\gamma_\mu] \} = \mathbb{P}_{\mathbb{F}_{q^r}}(\mathbb{F}_{q^m}) $, where $ [\gamma] = \{ \lambda \gamma \mid \lambda \in \mathbb{F}_{q^r}^* \} $ is the projective point associated to $ \gamma \in \mathbb{F}_{q^m}^* $ over $ \mathbb{F}_{q^r} $. Now since $ \gamma_i $ and $ \gamma_j $ are $ \mathbb{F}_{q^r} $-linearly independent if $ i \neq j $, then $ [\gamma_1], \ldots, [\gamma_\mu] $ are distinct projective points. Therefore they form the whole projective space if, and only if, there are $ (q^m-1) / (q^r-1) $ of them.
\end{proof}

This implies that Proposition \ref{prop doubly ext MSRD not MSRD} holds for $ 2 $-dimensional (full-length) linearized Reed--Solomon codes (the case $ r = m $ and $ \mu = \rho = 1 $, see \cite[Sec. 4.2]{generalMSRD}) and the more general family of MSRD codes obtained from Hamming codes given in \cite[Sec. 4.4]{generalMSRD}, which are the longest known $ 2 $-dimensional linear MSRD codes. In other words, those two families of $ 2 $-dimensional MSRD codes may not be doubly extended as in Proposition \ref{prop doubly ext MSRD not MSRD}. In the case $ r = 2 $, it was known that the latter family could not be doubly extended as in Proposition \ref{prop doubly ext MSRD not MSRD} since their number of blocks (the parameter $ g = (q-1) \mu $) attains the upper bound from \cite[Th. 6.12]{alberto-fundamental} since $ g = (q-1)(q^m-1) / (q^r-1) - 1 $ in this case. The fact that it may not be doubly extended for $ r \geq 3 $ is new.

\section{One-weight codes} \label{sec one-weight}

In this section, we give necessary and sufficient conditions for the doubly extended MSRD codes from Corollary \ref{cor doubly extended msrd codes} to be one-weight codes (or constant-weight codes), that is, such that all of their codewords have the same weight (thus equal to the minimum distance of the code). The next proposition recovers \cite[Th. 4.9]{neri-oneweight} for linearized Reed--Solomon codes by taking $ \mu = 1 $.

\begin{proposition} \label{prop one-weight characterization}
Let $ a_1, a_2, \ldots, a_{q-1} \in \mathbb{F}_{q^m}^* $ be such that $ N_{\mathbb{F}_{q^m}/\mathbb{F}_q}(a_i) \neq N_{\mathbb{F}_{q^m}/\mathbb{F}_q}(a_j) $ if $ i \neq j $. Let $ \boldsymbol\beta = (\beta_1, \beta_2, \ldots, \beta_{\mu r}) \in \mathbb{F}_{q^m}^{\mu r} $ and $ \mathcal{H}_i = \langle \beta_{(i-1)r+1} , \ldots, \beta_{ir} \rangle_{\mathbb{F}_q} \subseteq \mathbb{F}_{q^m} $, for $ i \in [\mu] $, satisfy Conditions 1 and 2 in Theorem \ref{th extended moore matrix is msrd sufficient}. Consider the extended sum-rank metric
$$ {\rm d}_e( (\mathbf{c}, c_{n+1}, c_{n+2}), (\mathbf{d}, d_{n+1}, d_{n+2})) = {\rm d}_{SR}(\mathbf{c},\mathbf{d}) + {\rm d}_H((c_{n+1}, c_{n+2}), (d_{n+1}, d_{n+2})), $$
for $ \mathbf{c} , \mathbf{d} \in \mathbb{F}_{q^m}^n $ and $ c_{n+1}, c_{n+2}, d_{n+1}, d_{n+2} \in \mathbb{F}_{q^m} $, where $ {\rm d}_{SR} $ denotes the sum-rank metric in $ \mathbb{F}_{q^m}^n $ over $ \mathbb{F}_q $ for the length partition $ (g,r) $, where $ g = (q-1) \mu $ and $ n = gr $. Then the extended $ 2 $-dimensional MSRD code $ \mathcal{C}_e $ with generator matrix 
$$ G_e = \left( \begin{array}{ccc|ccc|c|ccc|cc}
\beta_1 & \ldots & \beta_{\mu r} & \beta_1 & \ldots & \beta_{\mu r} & \ldots & \beta_1 & \ldots & \beta_{\mu r} & 1 & 0 \\
a_1 \beta_1^q & \ldots & a_1 \beta_{\mu r}^q & a_2 \beta_1^q & \ldots & a_2 \beta_{\mu r}^q & \ldots & a_{q-1} \beta_1^q & \ldots & a_{q-1} \beta_{\mu r}^q & 0 & 1
\end{array} \right) $$
is a one-weight code for $ {\rm d}_e $ if, and only if, $ \bigcup_{i=1}^\mu \mathcal{H}_i = \mathbb{F}_{q^m} $.
\end{proposition}
\begin{proof}
Let $ \mathbf{g}_1, \mathbf{g}_2 \in \mathbb{F}_{q^m}^{n+2} $ be the first and second rows of $ G_e $, respectively. Since $ {\rm d}_e(\mathcal{C}_e) = n+1 $, we need to show that $ {\rm wt}_e(\mathbf{g}_1 + \lambda \mathbf{g}_2) = n+1 $, for all $ \lambda \in \mathbb{F}_{q^m}^* $. Fix $ \lambda \in \mathbb{F}_{q^m}^* $. We need to show that there exist $ \lambda_1, \lambda_2, \ldots, \lambda_r \in \mathbb{F}_q $, not all zero, such that
$$ \sum_{k=1}^r \lambda_k \beta_{(j-1)r+k} + \lambda a_i \sum_{k=1}^r \lambda_k \beta_{(j-1)r+k}^q = 0, $$
for some $ j \in [\mu] $ and some $ i \in [q-1] $. Let $ \beta = \sum_{k=1}^r \lambda_k \beta_{(j-1)r+k} \in \mathcal{H}_j \setminus \{ 0 \} $. Then the equation above is simply $ -\lambda^{-1} = a_i \beta^{q-1} $. This is possible for all $ \lambda \in \mathbb{F}_{q^m}^* $ if, and only if, 
\begin{equation}
\mathbb{F}_{q^m}^* = \bigcup_{i=1}^{q-1} \left\lbrace a_i \beta^{q-1} \left| \beta \in \bigcup_{j=1}^\mu \mathcal{H}_j \setminus \{0 \} \right. \right\rbrace .
\label{eq union condition one-weight proof}
\end{equation}
Since $ \beta^{q-1} = \gamma^{q-1} $ holds for $ \beta,\gamma \in \mathbb{F}_{q^m}^* $ if, and only if, $ \beta/\gamma \in \mathbb{F}_q^* $, it is easy to see that (\ref{eq union condition one-weight proof}) holds if, and only if, $ \bigcup_{i=1}^\mu \mathcal{H}_i = \mathbb{F}_{q^m} $, and we are done.
\end{proof}

In the case where $ \boldsymbol\beta $ is constructed using field reduction as in Lemma \ref{lemma union condition for field reduction}, we see that the extended $ 2 $-dimensional MSRD code $ \mathcal{C}_e $ is a one-weight code for $ {\rm d}_e $ if, and only if, $ \mu = (q^m-1) / (q^r-1) $.

In other words, $ 2 $-dimensional doubly extended linearized Reed--Solomon codes and the doubly extended MSRD codes based on Hamming codes as in \cite[Sec. 4.4]{generalMSRD} are all one-weight codes for the extended metric $ {\rm d}_e $.

Finally, we show that triply extended MSRD codes are never one-weight codes for $ q = 2 $. Due to the results from Section \ref{sec triply extended}, we only consider the case where $ m $ is odd. Notice that in this case the vector $ \mathbf{a} $ is of length one and we may simply consider it as $ \mathbf{a} = (1) $.

\begin{proposition}
Let $ q = 2 $, let $ m \geq 3 $ be odd and set $ n = \mu r $ for positive integers $ \mu $ and $ r $. Let $ \boldsymbol\beta = (\beta_1, \beta_2, \ldots, \beta_{\mu r}) \in \mathbb{F}_{2^m}^{\mu r} $ satisfy Conditions 1 and 2 in Theorem \ref{th extended moore matrix is msrd sufficient}. The triply extended code $ \mathcal{C}_e \subseteq \mathbb{F}_{2^m}^{n + 3} $ with generator matrix
$$ G_e = \left( \begin{array}{cccc|ccc}
\beta_1 & \beta_2 & \ldots & \beta_{\mu r} & 1 & 0 & 0 \\
\beta_1^2 & \beta_2^2 & \ldots & \beta_{\mu r}^2 & 0 & 1 & 0 \\
\beta_1^4 & \beta_2^4 & \ldots & \beta_{\mu r}^4 & 0 & 0 & 1
\end{array} \right) \in \mathbb{F}_{2^m}^{3 \times (n + 3)} $$
is MSRD but not a one-weight code for the extended sum-rank metric
$$ {\rm d}_e( (\mathbf{c}, \mathbf{c}^\prime), (\mathbf{d}, \mathbf{d}^\prime)) = {\rm d}_{SR}(\mathbf{c},\mathbf{d}) + {\rm d}_H(\mathbf{c}^\prime, \mathbf{d}^\prime), $$
for $ \mathbf{c} , \mathbf{d} \in \mathbb{F}_{2^m}^n $ and $ \mathbf{c}^\prime, \mathbf{d}^\prime \in \mathbb{F}_{2^m}^3 $, where $ {\rm d}_{SR} $ denotes the sum-rank metric in $ \mathbb{F}_{2^m}^n $ over $ \mathbb{F}_2 $ for the length partition $ (\mu, r) $.
\end{proposition}
\begin{proof}
The fact that $ \mathcal{C}_e $ is MSRD for $ {\rm d}_e $ is Theorem \ref{th triply extended LRS}. Now, since $ {\rm d}_e(\mathcal{C}_e) = n-2 $, it is enough to show that there exists a codeword $ \mathbf{c} \in \mathcal{C}_e $ with $ {\rm wt}_e (\mathbf{c}) = n $. For $ \lambda, \nu \in \mathbb{F}_{2^m}^* $, let 
$$ \mathbf{c}_{\lambda,\nu} = ( \lambda \beta_1 + \nu \beta_1^2 + \beta_1^4, \ldots, \lambda \beta_{\mu r} + \nu \beta_{\mu r}^2 + \beta_{\mu r}^4, \lambda, \nu , 1 ) \in \mathcal{C}_e. $$
Since $ \lambda \neq 0 \neq \nu $, it holds that $ {\rm wt}_e (\mathbf{c}_{\lambda,\nu}) < n $ if, and only if, there exists an index $ i \in [\mu] $ and scalars $ \lambda_1, \lambda_2, \ldots, \lambda_r \in \mathbb{F}_2 $, not all zero, such that
$$ \sum_{j=1}^r \lambda_j \left( \lambda \beta_{(i-1)r+j} + \nu \beta_{(i-1)r+j}^2 + \beta_{(i-1)r+j}^4 \right) = 0. $$
By considering $ \beta = \sum_{j=1}^r \lambda_j \beta_{(i-1)r+j} \in \mathbb{F}_{q^m}^* $, we have that $ {\rm wt}_e (\mathbf{c}_{\lambda,\nu}) < n $ if, and only if, there exists an index $ i \in [\mu] $ and $ \beta \in \mathcal{H}_i = \langle \beta_{(i-1)r+1}, \beta_{(i-1)r+2}, \ldots, \beta_{ir} \rangle_{\mathbb{F}_2} \setminus \{ 0 \} $ such that $ \lambda \beta + \nu \beta^2 + \beta^4 = 0 $, that is, $ \beta^3 + \nu \beta + \lambda = 0 $. 

Now, since by \cite[Th. 3.25]{lidl} there are $ N_{2^m}(3) = (2^{3m} - 2^m)/3 > 2 \cdot 2^{2m} $ irreducible polynomials in $ \mathbb{F}_{2^m}[x] $, then there is at least one irreducible polynomial $ f = x^3 + a x^2 + b x + c \in \mathbb{F}_{2^m}[x] $ such that $ b \neq a^2 $ and $ b \neq 1 $. Furthermore, $ c \neq 0 $ since $ f $ is irreducible. Define $ g = f(x+a) = x^3 + (a^2+b)x + c(b+1) $, which is irreducible since so is $ f $. Let $ \nu = a^2+b $ and $ \lambda = c(b+1) $, which satisfy $ \lambda \neq 0 \neq \nu $. Since $ g $ is irreducible of degree $ 3 $, there is no $ \beta \in \mathbb{F}_{2^m} $ such that $ g(\beta) = \beta^3 + \nu \beta + \lambda = 0 $. In other words, the codeword $ \mathbf{c}_{\lambda,\nu} \in \mathcal{C}_e $ as above satisfies $ {\rm wt}_e (\mathbf{c}_{\lambda,\nu}) = n $, and we are done.
\end{proof}

\section*{Acknowledgement}

The author gratefully acknowledges the support from a Mar{\'i}a Zambrano contract by the University of Valladolid, Spain (Contract no. E-47-2022-0001486).

\footnotesize
 

\end{document}